\documentclass[aps,pra,showpacs,twocolumn,superscriptaddress, floatfix, nofootinbib]{revtex4-2}%

\usepackage[utf8]{inputenc}
\usepackage{epsfig}
\usepackage{url}
\usepackage{multirow}
\usepackage{makecell,booktabs} % for triple-line table
\usepackage{capt-of} % for captions of subtable / subfigure

\usepackage{diagbox} %用于表头斜线分割
\usepackage{graphicx}
\graphicspath{ {images/} }
\usepackage{tikz} % Enable plotting figures via TikZ package.

\usepackage{amsfonts}
\usepackage{amssymb}
\usepackage{amsmath}
\usepackage{amsthm}
\usepackage{verbatim}
\usepackage{listings}
\usepackage{float} % used to place figure with specifier [H]
\usepackage{bm} % bold math
\usepackage[qm,braket]{qcircuit} % used to plot quantum circuit and output kets and bras

\usepackage{mathtools}

\usepackage{color,xcolor}
\usepackage{hyperref}
\hypersetup{  colorlinks=true, linkcolor=blue, citecolor=red, urlcolor=blue  }
\usepackage[capitalize]{cleveref}
\usepackage{physics}

\renewcommand{\bf}[1]{\mathbf{#1}}

\newtheorem{theorem}{Theorem}
\newtheorem{corollary}{Corollary}[theorem]
\newtheorem{proposition}{Proposition}[theorem]
\newtheorem{lemma}{Lemma}

\newtheorem{definition}{Definition}

\usepackage{txfonts}
\DeclareMathOperator{\poly}{poly}
\newcommand{\fracp}[2]{\left( \frac{#1}{#2} \right)} % fraction with parentheses
\DeclareMathOperator{\Cov}{Cov}
\DeclareMathOperator{\Var}{Var}
\def\EE{\mathbb{E}}
\def\calU{\mathcal{U}}
\def\calF{\mathcal{F}}
\DeclareMathOperator{\Haar}{Haar}
\def\RR{\mathbb{R}}

\begin{document}

\title{ Generalized Cross-Entropy Benchmarking for Random Circuits with Ergodicity}

\author{Bin Cheng}
\thanks{These authors contributed equally}
\affiliation{Shenzhen Institute for Quantum Science and Engineering, and Department of Physics, Southern University of Science and Technology, Shenzhen 518055, China}
\affiliation{Centre for Quantum Technologies, National University of Singapore, Singapore}
\affiliation{Centre for Quantum Software and Information, Faculty of Engineering and Information Technology, University of Technology Sydney, NSW 2007, Australia}

\author{Fei Meng}
\thanks{These authors contributed equally}
\affiliation{Shenzhen Institute for Quantum Science and Engineering, and Department of Physics, Southern University of Science and Technology, Shenzhen 518055, China}
\affiliation{Department of Physics, City University of Hong Kong, Tat Chee Avenue, Kowloon, Hong Kong, China.}
\affiliation{QICI Quantum Information and Computation Initiative, Department of Computer Science, The University of Hong Kong, Pokfulam Road, Hong Kong SAR, China}

\author{Zhi-Jiong Zhang}
\affiliation{Shenzhen Institute for Quantum Science and Engineering, and Department of Physics, Southern University of Science and Technology, Shenzhen 518055, China}

\author{Man-Hong Yung}
\email{yung@sustech.edu.cn}
\affiliation{Shenzhen Institute for Quantum Science and Engineering, and Department of Physics, Southern University of Science and Technology, Shenzhen 518055, China}
\affiliation{International Quantum Academy, Shenzhen, 518048, China}
\affiliation{Guangdong Provincial Key Laboratory of Quantum Science and Engineering, Southern University of Science and Technology, Shenzhen, 518055, China}
\affiliation{Shenzhen Key Laboratory of Quantum Science and Engineering, Southern University of Science and Technology, Shenzhen 518055, China}

% Include the date command, but leave its argument blank.

%\date{}

%\begin{CJK*}{GBK}{kai}


%\baselineskip24pt


\begin{abstract}
Cross-entropy benchmarking is a central technique used to certify a quantum chip in recent experiments. To better understand its mathematical foundation and develop new benchmarking schemes, we introduce the concept of ergodicity to random circuit sampling and find that the Haar random quantum circuit satisfies an ergodicity condition---the average of certain types of post-processing function over the output bit strings is close to the average over the unitary ensemble. For noiseless random circuits, we prove that the ergodicity holds for polynomials of degree $t$ with positive coefficients and when the random circuits form a unitary $2t$-design. For strong enough noise, the ergodicity condition is violated. This suggests that ergodicity is a property that can be exploited to certify a quantum chip. 
We formulate the deviation of ergodicity as a measure for quantum chip benchmarking and show that it can be used to estimate the circuit fidelity for global depolarizing noise and weakly correlated noise. For a quadratic post-processing function, our framework recovers Google's result on estimating the circuit fidelity via linear cross-entropy benchmarking (XEB), and we give a sufficient condition on the noise model characterizing when such estimation is valid. Our results establish an interesting connection between ergodicity and noise in random circuits and provide new insights into designing quantum benchmarking schemes.

\end{abstract}

\maketitle

%%%%%%%%%%%%%%%%%%%%%%%%%%%%%%%%%%%%%%%%%%%%%%%%

%%%%%%Introduction

\section{Introduction} 

Recent advances in quantum technology offer significant promise but simultaneously impose challenging precision requirements on quantum device components. Certification of quantum devices~\cite{eisert2020quantum,kliesch2021theory,resch2021benchmarking,wang2022topologically} to ensure accurate operations and correct outputs is crucial to meaningful applications and to demonstrate quantum advantages, especially when current quantum devices are small-to-medium-sized and noisy, operating within the noisy intermediate-scale quantum (NISQ) regime~\cite{Preskill2018}.
For example, quantum computational supremacy---the superiority of near-term quantum computation without error correction compared to its classical counterpart~\cite{Preskill2012, harrow_quantum_2017, yung_quantum_2019} can only be established when the noise of the quantum device is sufficiently small. 
Protocols for demonstrating quantum computational supremacy includes boson sampling~\cite{Aaronson2010} and its variants~\cite{Lund2014,Hamilton2017,brod2019photonic,gao2022quantum}, instantaneous quantum polynomial-time (IQP) sampling~\cite{IQP10, IQP15}, and random circuit sampling (RCS)~\cite{boixo_characterizing_2018}. 
Random circuit sampling has emerged as a particularly promising approach, given its recent experimental implementations with over 50 qubits~\cite{arute_quantum_2019, wu_strong_2021, zhu_quantum_2021,decross_computational_2024} and the complexity-theoretic analysis of its classical hardness~\cite{Aaronson2017, mann_complexity_2017, hangleiter_anticoncentration_2018, harrow_approximate_2018, Bouland2019, movassagh_cayley_2019, arute_quantum_2019, dalzell_random_2020, morimae_sampling_2019}.

\begin{figure}[t]
    \centering
    \includegraphics[width = 90mm]{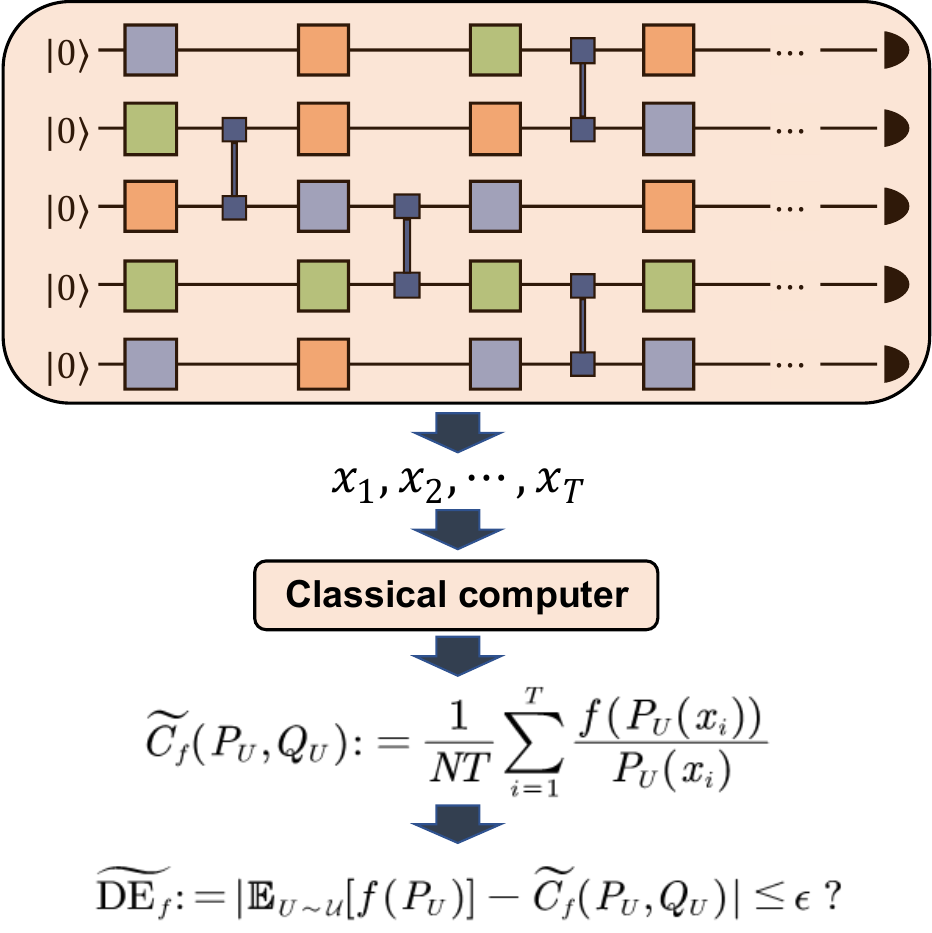}
    \caption{Basic procedure of our quantum chip benchmarking scheme. 
    Here, a random quantum circuit from a random circuit ensemble is sampled.
    Then, this circuit is repeatedly run and measured in the computational basis.
    The measurement outcomes are processed by a classical computer to get an estimate of the correlation between the ideal output distribution and the experimental output distribution.}
    \label{fig:benchmark_procedure}
\end{figure}

Demonstrating quantum computational supremacy relies on corresponding benchmarking schemes that measure the performance of the quantum device to ensure its correct function.
But benchmarking quantum devices in the supremacy regime is notoriously challenging~\cite{rinott_statistical_2020}.
For example,  state fidelity~\cite{Jozsa1994} measures the overlap between the observed (with noise) quantum state to the ideal (without noise) quantum state. 
Randomized benchmarking~\cite{Knill2008,wallman2014randomized,magesan2011scalable,helsen2022general} can be employed to evaluate the fidelity of individual quantum gates.
Process tomography~\cite{Poyatos1997,torlai2023quantum,white2022non,mohseni2008quantum} reveals the full details of the quantum device but is simultaneously the most resource-intense. 
However, all these methods become impractical for assessing the experimental performance of quantum chips at a large scale, due to the exponential growth in dimensionality~\cite{eisert2020quantum}. 
Moreover, in a typical random-circuit sampling experiment with 50 or more qubits, the majority of output bit strings occur only once even when millions of bit strings are sampled in a single experiment, rendering the reliable estimation of probabilities or calculation of expectation values for observables infeasible.

Benchmarking schemes with less resource cost in the supremacy regime can be achieved with less information gained or by imposing assumptions on the device ~\cite{eisert2020quantum}. For example,  fidelity gives much less information about the device but saves tremendous resources regarding measurements and sample complexity than full tomography~\cite{flammia2011direct,reich2013optimal}. Fidelity witness~\cite{gluza2018fidelity,aolita2014reliable} tests if the state fidelity passes a certain threshold, which is weaker than fidelity but further reduces the resource cost. A similar strategy to certify the quantum chip is to verify the presence of certain key properties, such as entanglement and entanglement propagation~\cite{jurcevic2014quasiparticle,audenaert2006correlations,courme2023manipulation}.  
Imposing assumptions can also reduce the resource cost for quantum certification. 
For example, a quantum simulator can be verified by validating its correct function in regimes where the system's behavior can be efficiently simulated by classical computers, with the assumption that such a performance of the quantum simulator extends to the classically intractable regime~\cite{trotzky2012probing,trotzky2010suppression}.
Similarly, the standard approach to certify a quantum device via testing each qubit and extending the trust to the whole chip assumes no global obstruction to the chip's performance. The cross-entropy benchmarking (XEB)~\cite{boixo_characterizing_2018,arute_quantum_2019} can be used to estimate circuit fidelity when the noise is assumed to be small and weakly correlated. 
While loopholes and spoofs are almost inevitable when the certification is based on assumptions~\cite{oh2023spoofing,gao2024limitations,suzuki2021experimental}, these protocols offer affordable solutions for benchmarking and certifying quantum devices, especially when the corresponding assumptions can be justified for specific experimental platforms and hardware.

We contribute to this line of research by proposing a new benchmarking scheme based on the ergodicity of random quantum circuits, a concept we borrow from statistical physics. Ergodicity is an important analytical tool of equilibrium statistical mechanics, which formulates the condition when the time average of an observable is the same (or sufficiently close) to the ensemble average~\cite{peters2019ergodicity,birkhoff1931proof,neumann1932proof,moore2015ergodic}. When the ergodicity condition holds, dynamical descriptions of a physical system can often be replaced with probabilistic distributions, which is much simpler since time dependence is eliminated and is the core technique used in equilibrium statistical mechanics~\cite{peters2019ergodicity,gibbs1902elementary}. Recently, there have been several works studying the behavior of noisy random quantum circuits by mapping the random quantum circuits into models in statistical mechanics~\cite{liu_benchmarking_2021, dalzell_random_2021,deshpande_tight_2021,  gao_limitations_2024}. Along with this research direction, we extend the concept of ergodicity to random circuit sampling, with the time average replaced by the average over output bitstrings.

Specifically, we formulate the ergodicity condition---the average of a function $f$ over a random circuit ensemble is close to the average over output bit strings of a single circuit instance from this ensemble.
Then, we prove that unitary $2t$-design~\cite{dankert_exact_2009, gross_evenly_2007} satisfies the ergodicity condition relative to polynomials with degree at most $t$ and positive coefficients. 
We also prove the ergodicity holds for the function $f(p)=p \ln p$, and $f(p) = \sum_i a_i p^{\theta_i}+b$ with positive real numbers $a_i$ and $\theta_i$ and real number $b$, when the unitaries are Haar random.
Based on this, we design a benchmarking scheme that quantifies the performance of quantum chips through the deviation of ergodicity.  Our benchmarking scheme only requires $T$ independent samples from the experiment, where $T=O(\poly(\log N))$ is significantly smaller than the Hilbert space dimension $N$, which makes it a practical tool for benchmarking in the supremacy regime.
We also analyze the relation between the deviation of ergodicity and fidelity under certain noise model assumptions. In particular, our framework recovers Google's result on estimating the circuit fidelity via linear cross-entropy benchmarking (XEB), and we give rigorous criteria on the noise model characterizing when such estimation is valid, and thus contributes to the research of technical aspects of XEB~\cite{boixo_characterizing_2018,arute_quantum_2019,aaronson_classical_2019, liu_benchmarking_2021, dalzell_random_2021,deshpande_tight_2021,  gao_limitations_2024}. 

Our results formulate a novel property of random circuit sampling, establish an interesting connection between ergodicity and noise in random circuits, and provide new insights into designing quantum benchmarking schemes.

\section{Ergodicity of Random Circuit Ensemble}
\label{sec:ergodic}

Before defining ergodicity, we briefly introduce the scheme of random circuit sampling  (RCS).
RCS is a task of sampling bit-strings from the output probability distribution of a (pseudo-)random quantum circuit~$U$ acting on the initial state $\ket{0^n}$, as illustrated in \cref{fig:benchmark_procedure}.
Here, $U \in \mathcal{U}$ is uniformly randomly sampled from some (possibly discrete) circuit ensemble $\mathcal{U}$, where the randomness is over the choices of local quantum gates in the quantum circuit.
RCS requires the quantum circuit to output a set of bit strings by performing computational basis measurements.
For each bit string $x \in \{ 0,1 \}^n$ of $n$ qubits, we will consider the ideal output distribution~$\{P_U(x)\}$ and the experimental output distribution~$\{Q_U(x)\}$, defined respectively as
\begin{align}
    P_U(x) &:= \left|\mel{x}{U}{0^n}\right|^2 \ , \\
    Q_U(x) &:= \mel{x}{\mathcal{E}_U\left(\op{0^n}{0^n}\right)}{x} \ .
\end{align}
Here, the quantum channel $\mathcal{E}_U$ models the experimental process for implementing circuit $U$. The output probability $P_U(x)$ can be regarded as a random variable when the unitary $U$ is randomly sampled from an ensemble $\calU$. If $U$ is sampled from Haar measure, then for a fixed bit string $x_0$, the ideal output probability $p = P_U(x_0)$ should approximately obey the Porter-Thomas distribution $\Pr(p) = N e^{-Np}$ (see \cref{sec:haar_output_prob} for an alternative proof), where the sample space is the group of $N \times N$ unitary operators $\mathbb{U}(N)$~\cite{boixo_characterizing_2018}. 
More practically, we will take the ensemble $\calU$ as a unitary $t$-design~\cite{dankert_exact_2009, gross_evenly_2007}, which is a finite set of unitaries replicating the statistical behavior of the Haar measure up to the $t$-th moment.
\begin{definition}[Unitary $t$-design]\label{def:unitary-t-design}
    A discrete circuit family $\calU$ of $N \times N$ unitary operators is called a unitary $t$-design if it satisfies,
    \begin{align}
        \frac{1}{|\calU|} \sum_{U \in \calU} U^{\otimes t} \otimes (U^{\dagger})^{\otimes t} = \int_{\mathbb{U}(N)} U^{\otimes t} \otimes (U^{\dagger})^{\otimes t} \dd{U} \ ,
    \end{align}
    where the integral is taken over the Haar measure on the unitary group $\mathbb{U}(N)$.
\end{definition}

The ergodicity of random quantum circuits is defined for a circuit ensemble $\calU$ and a smooth function $f(p)$ with $p = P_U(x)$, which states that the average of $f(P_U(x))$ over the outputs $x$ is close to the average over unitaries $U$. 
The mathematical definition of the ergodicity of random circuit ensembles is as follows.
\begin{definition}[Ergodicity]\label{def:ergodicity}
Let $U_0$ be an $n$-qubit quantum circuit sampled from an ensemble $\calU$ and $f: [0, 1] \to \RR$ be some smooth function. Then, the ensemble $\calU$ is said to satisfy ergodicity relative to the function $f$ if for an arbitrary bit string $x_0 \in \{ 0,1 \}^n$, we have
\begin{align}\label{eq:ergodic}
\left| \mathbb{E}_U \left[ f(P_U(x_0)) \right] - \frac{1}{N}\sum_{x \in \{0,1\}^n} f(P_{U_0}(x)) \right| = O\fracp{\sigma_f}{\sqrt{N}} \ ,
\end{align}
with high probability, where $\sigma_f := \sqrt{\mathbb{E}_U\left[(f-\mathbb{E}_U(f))^2\right]}$ is the standard deviation of the random variable $f(P_U(x_0))$ over $\calU$, $N = 2^n$ and
\begin{align}\label{eq:ensemble_average_def}
    \mathbb{E}_U [f(P_U(x_0)] := \frac{1}{|\mathcal{U}|} \sum_{U \in \calU} f(P_U(x_0)) \ ,
\end{align}
is the ensemble average over $U$.
\end{definition}

As a simple example, let $\calU$ be the 1-qubit Pauli operators and let $f(p) = p$.
Then, the second term in \cref{eq:ergodic} is $1/2$.
As for the first term, direct calculation gives $\EE_U [P_U(x)] = \frac{1}{4}\sum_{P \in \{I, X, Y, Z \}} \mel{x}{P^{\dagger}}{0} \mel{0}{P}{x} = 1/2$.
Thus, the set of 1-qubit Pauli operators satisfies ergodicity relative to the linear function $f(p) = p$.
In fact, it is not so hard to generalize the calculation to $n$-qubit Pauli operators and show that it also satisfies the ergodicty on the linear function.
Moreover, as indicated by \cref{thm:ergodicity_polynomials}, if one chooses higher-order functions, then the circuit ensemble should also be more close to Haar measure.

\cref{def:ergodicity} takes a similar form as the ergodic theorem~\cite{moore2015ergodic}, which states that the time average of a function $f$ of a dynamical system is equal to the ensemble average, except that the average over time is replaced by the average over output bit strings $x$. 
The above formulation of ergodicity is also a statement of concentration of measure~\cite{muller2011concentration,mcclean2018barren,ledoux_concentration_2005}, but it is different from Levy's lemma~\cite{milman1986asymptotic,ledoux2001concentration}, another statement of concentration of measure commonly used in quantum information science.
However, it does not appear that one can derive the ergodicity condition from Levy's lemma.
A detailed discussion of its connection to Levy's lemma is presented in \cref{sec:levy_lemma_discussion}.

Our main technical contribution is the following theorem, which characterizes a family of circuit ensemble $\calU$ and functions $f$, for which the ergodicity holds. From now on, we omit the subscript of $U_0$ when the context is clear.
\begin{theorem}[Ergodicity relative to polynomials] \label{thm:ergodicity_polynomials}
    A unitary $2t$-design $\calU$ satisfies ergodicity relative to the function $f(p) = \sum_{i=1}^t a_i p^i + b$ with $a_i$ non-negative and $b$ a constant number. More explicitly, for such a function $f(p)$, we have
\begin{equation}\label{eq:ergodicity_with_confidence_interval}
    \Pr(\left| \mathbb{E}_U \left[ f(P_U(x_0)) \right] - \frac{1}{N}\sum_{x \in \{0,1\}^n} f(P_{U}(x)) \right|\geq \alpha \frac{\sigma_f}{\sqrt{N}})\leq \frac{1}{\alpha^2} \, ,
\end{equation}
holds for any positive constant $\alpha$ and $U \in \calU$. 
Moreover, when the unitary ensemble is the Haar measure, the ergodicity condition (\cref{eq:ergodicity_with_confidence_interval}) holds for $f(p) = \sum_i a_i p^{\theta_i} + b$, where $\theta_i$ and $a_i$ are positive real numbers.
\end{theorem}

We give some remarks before presenting the proof.
Such a polynomial $f$ with a non-negative coefficient is not the only type of functions for which the ergodicity holds. We prove that $f(p)=p \ln p$ also satisfies ergodicity for Haar random unitaries in \cref{thm:ergodicity_log}, which is also numerically illustrated in \cref{fig:numeric_RUC}.  
More investigations on characterizing which types of functions satisfy the ergodicity are left for future study. 
Characterizing the full family of functions for which the ergodicity condition holds is an interesting yet challenging mathematical question, which can help us understand the universality of the ergodicity condition. However, for practical use in quantum benchmarking, the established ergodicity for the polynomial function $f(p) = \sum_i^t a_i p^i + b$ and the logarithmic function $f(p) = p \ln p$ already covers the most important cases of cross-entropy benchmarking. 
Specifically, we will show in \cref{sec:general_framework} that the former can be used to reproduce the linear cross-entropy, while the latter can be used to reproduce the logarithmic cross-entropy---both employed as figures of merit in Google's experiment~\cite{arute_quantum_2019}. Nevertheless, finding additional types of functions for which the ergodicity condition holds could potentially broaden the toolkit for quantum benchmarking.

Second, $\alpha$ is a parameter that appears in both the distance and the probability, which specifies the difficulty to satisfy the ergodicity. For a given unitary $U$, we say that the ergodicity is satisfied with respect to $\alpha$ if $\left| \mathbb{E}_U \left[ f(P_U(x_0)) \right] - \frac{1}{N}\sum_{x \in \{0,1\}^n} f(P_{U}(x)) \right| \leq \alpha \frac{\sigma_f}{\sqrt{N}}$ holds.
Larger $\alpha$ means the deviation distance allowed is larger, while at the same time the corresponding probability for having such a large deviation is small. Thus picking a a larger $\alpha$ means to have higher probability in satisfying the ergodicity with respect to $\alpha$.
One should choose a suitable $\alpha$ so that the two averages are close to each other with high probability.
For example, in order to let the deviation of ergodicity holds with probability at least $0.99$, one should pick $\alpha = 10$ and it means that the two averages in the formulation of ergodicity (\cref{eq:ergodic}) will be allowed to be within distance $10\sigma_f/\sqrt{N}$. 
When a deviation larger than $10\sigma_f/\sqrt{N}$ is observed from the experiment, the ergodicity condition with respect to $\alpha=10$ is violated. Therefore, our theorem guarantees that, for a perfect quantum device, the probability for violating the ergodicity with respect to $\alpha$ for the function $f(p)$ is vanishingly small for large $\alpha$.

\begin{proof}[Proof of \cref{thm:ergodicity_polynomials}]
The proof is based on Chebyshev's inequality, which states that given a random variable $X$ with expectation $\mu$ and variance $\sigma^2$, we have
\begin{align}
    \Pr( |X - \mu| \geq a ) \leq \frac{\sigma^2}{a^2} \ .
\end{align}
Let $X = \frac{1}{N}\sum_x f(P_U(x))$ be a random variable, then its expectation value is $\mu = \mathbb{E}_U \left[ f(P_U) \right]$ and its variance is
\begin{align}
    \sigma^2 = \frac{1}{N^2} \sum_{x, y \in \{0,1\}^n} \Cov \left( f(P_U(x)), f(P_U(y)) \right) \ ,
\end{align}
where the covariance over $\calU$ is defined by,
\begin{align}
    & \Cov \left( f(P_U(x)), f(P_U(y)) \right) \notag \\
    := & \EE_U 
    \left[ f(P_U(x)) f(P_U(y)) \right] - \EE_U \left[ f(P_U(x)) \right] \EE_U \left[ f(P_U(y)) \right] \ .
\end{align}
By definition, we have
\begin{equation}
    \sigma^2 = \frac{1}{N} \sigma_f^2 +  \frac{1}{N^2} \sum_{x, y \in \{0,1\}^n,\, x \neq y} \Cov \left( f(P_U(x)), f(P_U(y)) \right) \ .
\end{equation}
The following \cref{lemma:covariance_over_U} implies that $\sigma^2 < \sigma_f^2 /N$ for (1) $f(p) = \sum_i a_i p^{\theta_i} + b$ with positive real numbers $a_i$ and $\theta_i$ and constant $b$ when the unitary ensemble is the Haar measure, or (2) $f(p) = \sum_{i=1}^t a_i p^i + b$ with $a_i$ non-negative and $b$ a constant number when the unitary ensemble is a unitary $2t$-design.
Now, set $a = \alpha \sigma_f/\sqrt{N}$.
According to Chebyshev's inequality, we have $\Pr(|X - \mu| \geq a) \leq \frac{1}{\alpha^2}$.
That is, the probability that the random variable $X$ deviates from its mean $\mu$ by a value at least $a$ is at most $1/\alpha^2$.
This matches the statement of \cref{def:ergodicity} and proves \cref{eq:ergodicity_with_confidence_interval}.
\end{proof}

\begin{lemma}\label{lemma:covariance_over_U}
Let $q_1$ and $q_2$ be two positive real numbers and let $x \neq y$. Then, for $N\times N$ Haar-distributed unitaries, we have,
\begin{align}
    & \Cov \left( P^{q_1}_U(x), P^{q_2}_U(y) \right) \notag \\ 
    = & \Gamma(q_1+1) \Gamma(q_2+1) \left[\frac{\Gamma(N)}{\Gamma(q_1+q_2+N)} - \frac{\Gamma(N)\Gamma(N)}{\Gamma(q_1+N)\Gamma(q_2+N)}\right] \notag \\
    < & 0 \ ,
\end{align}
where $\Gamma(m) = (m - 1)!$ is the Gamma function.
This implies that for any polynomial $f(p) = \sum_{i=1}^t a_i p^i + b$ with $a_i$ non-negative, we have,
\begin{align}
    \Cov \left( f(P_U(x)), f(P_U(y)) \right) < 0 \ .
\end{align}
Moreover, the Haar random ensemble can be relaxed to unitary $2t$-design , for such a choice of $f$ whose polynomial degree is $t$.
\end{lemma}

The proof to this lemma is based on calculating the corresponding Haar integrals~\cite{collins_integration_2006}, and we leave it in \cref{sec:analysis_random_unitary_condition}.
Here, we remark that for a polynomial $f(p)$ with degree $t$, the reason why the Haar random ensemble can be relaxed to a unitary $2t$-design~\cite{dankert_exact_2009, gross_evenly_2007}, is because the covariance $\Cov \left( f(P_U(x)), f(P_U(y)) \right)$ will have the same value for both Haar measure and unitary $2t$-design. 
By \cref{def:unitary-t-design}, computing statistical quantities up to $2t$-th moment over a unitary $2t$-design cannot be distinguished from computing them over Haar measure of the unitary group, and for a degree-$t$ polynomial, the covariance is a $2t$-th moment, since we have $U$ and $U^{\dagger}$ appear $2t$ times in its expression.
Finally, we would like to remark that, when the average is taken with respect to the Haar measure, the left- and right-translation invariance of the Haar measure eliminates the dependence on $x_0$. 
Based on the above discussion, we will omit $x_0$ and simply write $\EE_U[f(P_U)]$ in the remaining text when the context is clear.

\cref{thm:ergodicity_polynomials} establishes ergodicity relative to polynomials.
Using \cref{lemma:covariance_over_U}, one can also prove ergodicity relative logarithmic functions. 
In particular, we are interested in the function $f(p) = p \ln{p}$.
Such a function finds applications in the logarithmic cross-entropy benchmarking~\cite{boixo_characterizing_2018}.
We will also use the replica trick, a technique from statistical physics, stating that $\ln p = \lim_{i \rightarrow 0} \frac{p^i-1}{i}$.
With the replica trick and \cref{lemma:covariance_over_U}, we have the following theorem, whose proof is presented in \cref{appendix:ergodicity_proof_plogp}. 
\begin{theorem}[Ergodicity relative to logarithm] \label{thm:ergodicity_log}
    A Haar random unitary ensemble $\calU$ satisfies ergodicity relative to the function $f(p) = p \ln p$. More explicitly, for such a function $f(p)$, we have
\begin{equation}\label{eq:ergodicity_with_confidence_interval_plogp}
    \Pr(\left| \mathbb{E}_U \left[ f(P_U(x_0)) \right] - \frac{1}{N}\sum_{x \in \{0,1\}^n} f(P_U(x)) \right|\geq \alpha \frac{\sigma_f}{\sqrt{N}})\leq \frac{1}{\alpha^2} \, .
\end{equation}
\end{theorem}

\medskip
\textbf{Numerical support.}
For the case $f(p) = p$, we can apply \cref{lemma:covariance_over_U} to obtain
\begin{align}
\mathbb{E}_U \left[ P_U(x) P_U(y) \right] &= \frac{1}{N} \frac{1}{N+1} \\
\Cov(P_U(x), P_U(y)) &= -\frac{1}{N^2} \frac{1}{N+1} \ .
\end{align}
We perform numerical experiments to demonstrate the scaling of the expectation $\EE_U[P_U(x)]$, the second moment $\EE_U [P_U(x) P_U(y)]$ and the covariance $\Cov(P_U(x), P_U(y))$, which is relevant to \cref{def:ergodicity} for the case $f(p) = p$, and the results are shown in \cref{fig:conjectures}.
On the other hand, in \cref{fig:numeric_RUC}, we provide an explicit demonstration of the ergodicity condition for $f(p) = -\ln{p}$, which is not covered by \cref{thm:ergodicity_polynomials} { and \cref{thm:ergodicity_log}}.
Here, $|\mathrm{Error}|$ is the absolute difference between  $\mathbb{E}_U \left[ f(P_U) \right]$ and $N^{-1}\sum_{x}f(P_U(x))$ (the blue line). We also plot the scaling of error $O(\sigma_f/\sqrt{N})$ as a reference (the orange dashed line).

\begin{figure}[t]
    \centering
    \includegraphics[width = 90mm]{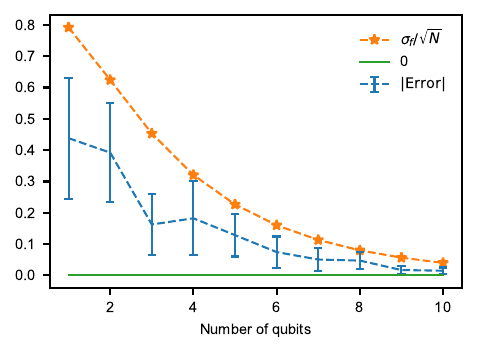}
    \caption{Numerical evidence of the ergdicity condition relative to $f(p) = -\ln{p}$.
    For each particular number of qubit $n$, 10 different instances of $U$ are randomly chosen for the calculation of $N^{-1}\sum_{x}f(P_U(x))$, while $\mathbb{E}_U[f(P_U)] = \int f(p) \Pr(p) \dd{p}$ is calculated as if $U$ is Haar random; 
    $|\text{Error}|$ is the left-hand side of \cref{eq:ergodic}.}
    \label{fig:numeric_RUC}
\end{figure}

\section{Generalized cross-entropy benchmarking}

\begin{figure*}[th]
    \centering
    \includegraphics[width = 18cm]{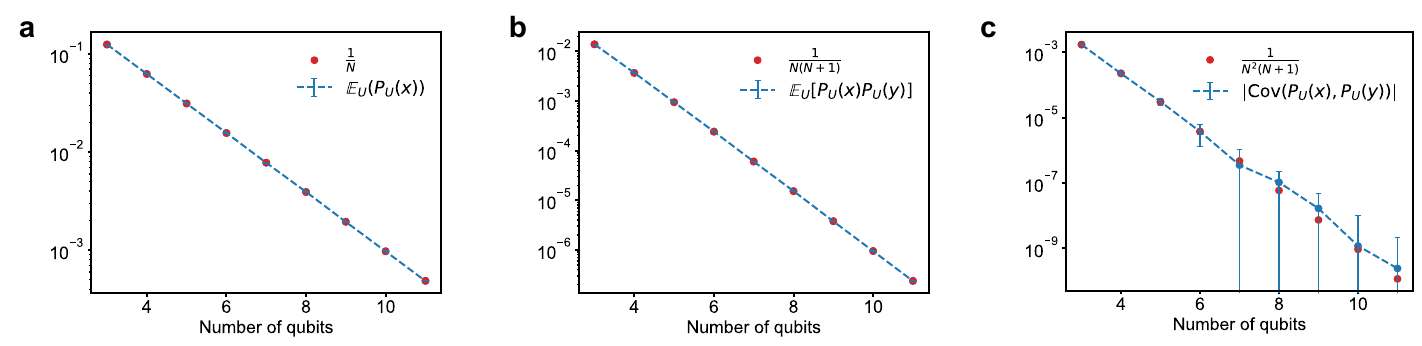}
    \caption{ The scaling of \textbf{(a)} the expectation $\EE_U [P_U(x)]$, \textbf{(b)} the second moment $\EE_U [ P_U(x) P_U(y)]$, and \textbf{(c)} the covariance $\Cov( P_U(x), P_U(y) )$. }
    \label{fig:conjectures}
\end{figure*}

\subsection{General framework}\label{sec:general_framework}

We propose a benchmarking scheme for quantum chips based on the ergodicity of random quantum circuits, which is a generalization of cross-entropy benchmarking~\cite{boixo_characterizing_2018,arute_quantum_2019}. 
The intuition comes from the observation that the ergodicity condition holds for noiseless chips and is violated when the noise level is high. Therefore, the deviation from the ergodicity condition can be used to quantify the noise level of a quantum chip. 
Moreover, we show that the deviation of ergodicity gives a linear cross entropy estimator with a quadratic scheme function $f(p)=p^2$ {and it gives the logarithm cross entropy estimator with $f(p)=p \ln p$}~\cite{arute_quantum_2019}.

We define the deviation of ergodicity as,
\begin{equation}
   \mathrm{DE}_f:= \abs{\EE_U [f(P_U(x))] -  C_f(P_U, Q_U)} \, ,
\end{equation}
where $C_f(P_U, Q_U)$ is a measure of the correlation between the ideal output distribution $\{ P_U(x) \}$ and the experimental distribution $\{ Q_U(x) \}$,
\begin{align}
    C_f(P_U, Q_U) := \frac{1}{N} \sum_{x \in \{0, 1\}^n} \frac{f(P_U(x))}{P_U(x)} Q_U(x) \ .
\end{align}
It recovers the second term in \cref{eq:ergodic} when $Q_U(x) = P_U(x)$. The idea behind this definition is that \cref{thm:ergodicity_polynomials} guarantees that the deviation is small for the noiseless circuit, and a large deviation indicates a high noise level.
Note that $C_f(P_U, Q_U)$ recovers the linear cross entropy when $f(P_U(x)) = P^2_U(x)$ {and recovers the logarithmic cross entropy when $f(P_U(x)) = P_U(x) \ln P_U(x)$}, up to rescaling and shift.

To evaluate the deviation of ergodicity, the first term $\EE_U [f(P_U(x))]$ can be calculated analytically using the Haar integral, if the circuit ensemble is chosen properly.
The second term can be estimated from experiments.
For a smooth function $f(p)$, if one takes $p = P_U(x_0)$ for some bit string $x_0$, then its average over the Haar-distributed $U$ can be evaluated by
\begin{align}\label{eq:ensemble_average_f}
    \mathbb{E}_{U \in \Haar} [f(P_U(x_0))] = \int_0^1 f(p) \Pr(p) \dd{p} \ ,
\end{align}
where $\Pr(p)$ is the probability density of the Porter-Thomas distribution~\cite{boixo_characterizing_2018} (see also \cref{sec:haar_output_prob}).
In particular, for $f(p)=p^i$ with $i$ being an positive integer, we computed in \cref{appendix:calculationforpi} that
\begin{equation}\label{eq:E_Upolynomial}
\mathbb{E}_{U \in \Haar} [P^i_U(x_0))] = \frac{i!}{N^i} + e^{-\Omega(N)} \, .
\end{equation}

As for the correlation term $C_f(P_U, Q_U)$, it can be estimated from experimental data as follows. (i) Sample $\{x_1, x_2, \cdots, x_T\}$ from the experimental output distribution $\{ Q_U(x) \}$.  (ii) For each sample $x_i$, compute its ideal output probability $P_U(x_i)$. (iii)  Compute the sample average
\begin{align}\label{eq:C_f_estimate}
    \widetilde{C}_f(P_U, Q_U) := \frac{1}{NT} \sum_{i = 1}^T \frac{f(P_U(x_i))}{P_U(x_i)}\ ,
\end{align}
which is an unbiased estimator to $C_f(P_U, Q_U)$.

To bound the estimation error of $\widetilde{C}_f(P_U, Q_U)$, let $g(x) = \frac{f(P_U(x))}{N P_U(x)}$ and we have the expectation of $g(x)$ over the experimental output distribution $\{ Q_U(x) \}$ is given by
\begin{align}
    \expval{g}_x := \sum_{x \in \{0, 1\}^n} \frac{f(P_U(x))}{N P_U(x)} Q_U(x) 
    =  C_f(P_U, Q_U) \, .
\end{align}
Then, the error of the estimation is given by the Chebyshev's inequality,
\begin{equation}\label{eq:exp_estimate}
    \left| \widetilde{C}_f(P_U, Q_U) - {C}_f(P_U, Q_U) \right| = O \left( \sqrt{\frac{\Var_{Q_U}[g]}{T}} \right) \ ,
\end{equation}
with high probability, where $\Var_{Q_U}[g]:=\ev*{g^2}_x - \ev{g}_x^2$ is the variance of $g(x)$ over the distribution $\left\{ Q_U(x) \right\}$.
Therefore, if the number of samples $T$ is large enough, then the sample mean will approximate the expectation $\ev{g}_x$ well. In practice, $T$ can be chosen as a polynomial in $\log N$, which is significantly smaller than the Hilbert dimension $N$. This scaling reduces significant sampling costs in benchmarking the chip.  In later case studies, one can see that that $T=O(\mathrm{poly}(\log N))$ scaling is sufficient for the benchmarking scheme.

Now, we denote the experimental estimation for the deviation of ergodicity by
\begin{equation}\label{eq:experiment_estimation_DoE}
   \widetilde{\mathrm{DE}}_f:= \abs{\EE_U [f(P_U)] -  \widetilde{C}_f(P_U, Q_U)} \ ,
\end{equation}
and use it to measure the quantum chip's performance. 
First, we choose the circuit ensemble and an appropriate scheme function $f$.
As indicated by \cref{thm:ergodicity_polynomials}, if the scheme function is a degree-$t$ polynomial, then the circuit ensemble $\calU$ should be chosen to be a unitary $2t$-design. 
Second, the term $\EE_U [f(P_U)]$ is computed analytically, using the property of unitary $2t$-design, i.e., $\EE_U [f(P_U)] = \EE_{U \in \Haar} [f(P_U)]$.
Finally, we experimentally estimate $\widetilde{C}_f(P_U, Q_U)$ using \cref{eq:C_f_estimate}, which gives $\widetilde{\mathrm{DE}}_f$.
The whole procedure is shown in \cref{fig:benchmark_procedure}.

\subsection{Case studies}
\label{subsec:case_studies}

\paragraph{Global depolarizing noise.}
As a toy model, we analyze the deviation of ergodicity for a quantum chip with global depolarizing noise. Under such a model, the output state $\rho_U$ of the quantum chip is given by,
\begin{equation}
    \rho_U = F \ket{\psi_U}\bra{\psi_U} + (1-F) \frac{\mathbb{I}}{N} \, ,
\end{equation}
where $\ket{\psi_U} = U \ket{0}^{\otimes n}$ is the ideal output of the noiseless quantum circuit $U$ and $F$ is equal to the circuit fidelity $\mel{\psi_U}{\rho_U}{\psi_U}$ up to an additive factor at most $1/N$. 
This noise model is also known as the global white noise in Ref.~\cite{dalzell_random_2021}.
Then, the experimental probability is given by
\begin{equation}
    Q_U(x) = F P_U(x) + (1-F) \frac{1}{N} \, .
\end{equation}
Thus, for $f(p) = N^i p^i$, we have $\ev{g}_x = F N^{i-1}  \sum_x P_U(x)^i + (1-F) N^{i-2} \sum_x P_U(x)^{i-1}$, which, according to \cref{thm:ergodicity_polynomials} and \cref{eq:E_Upolynomial}, gives 
\begin{equation}
    \ev{g}_x =F i! + (1-F) (i-1)! + O \left(\frac{1}{\sqrt{N}}\right) \, .
\end{equation} 
A more detailed derivation can be found in \cref{appendix:exampleofDeviationErgodicityforCompletely_Noisy_Chips}.
Similarly, $\ev*{g^2}_x - \ev{g}_x^2 = O(1)$. 
Then, from \cref{eq:exp_estimate}, the estimation of the correlation is equal to
\begin{equation}
\widetilde{C}_f(P_U, Q_U)= F i! + (1-F) (i-1)! \pm O\left( \frac{1}{\sqrt{T}} \right) \, ,
\end{equation}
with high probability.
Next, by substituting $\widetilde{C}_f(P_U, Q_U)$ and \cref{eq:E_Upolynomial} into \cref{eq:experiment_estimation_DoE}, we obtain
\begin{equation}
    \widetilde{\mathrm{DE}}_f = (1-F) (i-1)! (i-1) \pm O\left( \frac{1}{\sqrt{T}} \right)\, .
\end{equation}
Finally, rearranging this expression allows us to estimate the fidelity using the observed deviation of ergodicity from the experiment:
\begin{equation}
  F =1- \frac{1}{(i-1)!(i-1)} \, \widetilde{\mathrm{DE}}_f \pm O\left( \frac{1}{\sqrt{T}} \right) \, .
\end{equation}
Therefore, for the global depolarizing noise, the larger the observed deviation of ergodicity, the smaller the circuit fidelity.

\medskip
\paragraph{Weak correlated noise.}
Our framework recovers Google's results on linear cross-entropy benchmarking with when $f(p) = N^2 p^2$~\cite{arute_quantum_2019}. 
For a random quantum circuit $U$, the experimental output state can be decomposed into the following form,
\begin{align}
    \rho_U = F \dyad{\psi_U} + (1 - F) \chi_{U} \ , \label{eq:noise}
\end{align}
where $\ket{\psi_U} := U \ket{0^n}$ is the ideal output state,  $F := \mel{\psi_U}{\rho_U}{\psi_U}$ is the circuit fidelity, and $\chi_U$ is an operator with unit trace that will be used to characterize the experimental noise. For the output state $\rho_U = F \dyad{\psi_U} + (1 - F) \chi_U$, the experimental output distribution $Q_U(x):=\mel{x}{\rho_U}{x}$ can be written as
\begin{equation}\label{eq:noisy_distribution}
    Q_U(x) = F P_U(x) + (1-F) \chi_U(x) \ , 
\end{equation}
where $\chi_U(x):=\mel{x}{\chi_U}{x}$ is the quasiprobability associated with the error part  $\chi_U$.
Here, $\chi_U(x)$ is a quasiprobability because $\chi_U$ may not be positive. We prove the following proposition.

\begin{proposition}[Estimating circuit fidelity for weakly correlated noise]\label{theorem:estimation_circuit_fidelity}
For $f(p)=N^2 p^2$, assume the quasiprobability $\chi_U(x_0)$ obtained from the noise $\chi_U$ is weakly correlated with $P_U(x_0)$ for any fixed $x_0$ in the sense that
\begin{equation}\label{eq:independence_assumption}
    \mathbb{E}_U [\chi_U f(P_U)] = \mathbb{E}_U [\chi_U] \mathbb{E}_U [f(P_U)] + O\left( \frac{1}{N \sqrt{N}} \right) \ ,
\end{equation}
where the dependence on $x_0$ is omitted,
and the noise satisfies $\EE_U [\chi_U(x_0)] =\frac{1}{N}$.
Then, the deviation of ergodicity gives an estimation of the circuit fidelity, 
\begin{equation}
    F = 1- \widetilde{\mathrm{DE}}_f  \pm O\left( \frac{1}{\sqrt{T}} \right) \, .
\end{equation}
\end{proposition}

\begin{proof}
The noise distribution \cref{eq:noisy_distribution} and the linearity of $ C_f(P_U, Q_U)$ in $\{Q_U\}$ implies that
\begin{equation}
        C_f(P_U, Q_U) =F  C_f(P_U, P_U) + (1-F)   C_f(P_U, \chi_U) \, .
\end{equation}
The ergodicity \cref{thm:ergodicity_polynomials} establishes that
\begin{equation}
    C_f(P_U, P_U) =\mathbb{E}_{U} [N^2 P^2_U(x_0))]\pm O\left( \frac{1}{\sqrt{N}} \right) = 2 \pm O\left( \frac{1}{\sqrt{N}} \right) \, .
\end{equation}
Because of the weak correlation assumed, we have
\begin{equation}
     \EE_U [C_f(P_U, \chi_U)] =\frac{1}{N}\sum_{x \in \{0,1\}^n} \EE_U[N^2 P_U]\EE_U[\chi_U(x)] + O \left(\frac{1}{N \sqrt{N}} \right)  \, .
\end{equation}
Because $\EE_U [\chi_U(x)] = \frac{1}{N}$ for all $x$, then 
\begin{equation}
    \EE_U [C_f(P_U, \chi_U)] = 1 + O \left(\frac{1}{N\sqrt{N}} \right) \, .
\end{equation}
It is direct to verify that $C_f(P_U, \chi_U)$ is a Lipschitz function on $U$ with Lipschitz constant $\eta \geq \frac{N}{\sqrt{2}}$. Define 
\begin{equation}
    \delta = \abs{C_f(P_U, \chi_U) - \EE_U [C_f(P_U, \chi_U)]}
\end{equation} 
and by Levy's lemma (see \cref{sec:levy_lemma_discussion}), we have
\begin{equation}
   \Pr[ \delta \geq \frac{1}{\sqrt{N}}] \leq 2 \exp[-C(2N+1)N] \, ,
\end{equation}
where $C$ is some constant. Therefore, typically (with high probability) we have
\begin{equation}
    C_f(P_U, \chi_U)=\EE_U [C_f(P_U, \chi_U)] \pm O\left( \frac{1}{\sqrt{N}} \right) = 1 + O\left( \frac{1}{\sqrt{N}} \right) \, ,
\end{equation}
for any $U$. This implies that, typically, we have
\begin{equation}\label{eq:observedC_f_weak_noise}
\widetilde{C}_f(P_U, Q_U) = (1+F) \pm O\left( \frac{1}{\sqrt{T}} \right), 
\end{equation}
and in terms of the derivation of ergodicity,
\begin{equation}
    \widetilde{\mathrm{DE}}_f =(1-F) \pm O\left( \frac{1}{\sqrt{T}} \right) \, ,
\end{equation}
which concludes the proof.
\end{proof}

\begin{figure}[t]
    \centering
    \includegraphics[width=90mm]{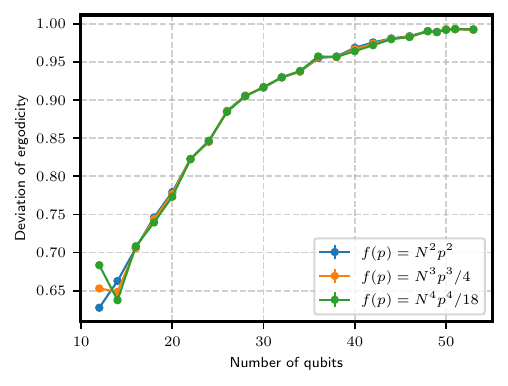}
    \caption{\textbf{Deviation of ergodicity $\widetilde{\mathrm{DE}}_{f}$ with various scheme functions $f$ for Google's Sycamore quantum chip~\cite{arute_quantum_2019} operating with different number of qubits $n$.} We evaluate the deviation of ergodicity of Google's quantum supremacy experiment based on their experimental data~\cite{arute_quantum_2019}. Here $N = 2^n$. Error bar is too small to be visible.}
\label{fig:DE_vs_number_of_qubits_of_Google}
\end{figure}

The above proposition recovers Google's claim that the linear XEB fidelity can be used to estimate the circuit fidelity when the noise is weakly correlated. Specifically, the linear XEB fidelity is defined as
\begin{equation}
    \calF_{\mathrm{XEB}} =  \frac{N}{T} \sum_{i = 1}^T P_U(x_i) -1 \,,
\end{equation}
where $x_i$ ranges over $T$ experimental samples.
Note that when $f(p) = N^2 p^2$, the correlation function $C_f(P_U, Q_U)$ reduces to the linear XEB fidelity,
\begin{equation}
    \widetilde{C}_f(P_U, Q_U) = \calF_{\mathrm{XEB}}+1 \ .
\end{equation}
Thus \cref{eq:observedC_f_weak_noise} is equivalent to
\begin{equation}
    \calF_{\mathrm{XEB}} = F \pm O\left( \frac{1}{\sqrt{T}} \right) \, ,
\end{equation}
which establishes the equivalence of the circuit fidelity to the linear XEB fidelity.

\paragraph{General noise model.}  
An important question is whether the deviation to ergodicity can reliably serve as a proxy for circuit fidelity under more general noise models. In Ref.~\cite{dalzell_random_2021}, it was rigorously shown that for weak and unital local noise (e.g., depolarizing, dephasing, or Pauli noise), random quantum circuits scramble local errors into an effective global depolarizing noise. This scrambling property justifies the use of deviation to ergodicity to estimate fidelity when the noise in real-world quantum devices is weak, local, and unital. Many experimental noise models, such as depolarizing and dephasing noise, fall under this unital and local noise framework. As such, the global depolarizing noise becomes a valid approximation, ensuring the reliability of estimating circuit fidelity using the deviation to ergodicity in these scenarios.

However, the global depolarizing noise approximation can break down in the presence of strongly correlated noise or coherent noise~\cite{dalzell_random_2021}, where the deviation to ergodicity may fail to estimate circuit fidelity. Furthermore, in adversarial scenarios, where noise can be deliberately engineered to exploit vulnerabilities in benchmarking schemes, the deviation to ergodicity can fail as a reliable proxy for fidelity. Supporting evidence from Ref.~\cite{gao_limitations_2024} highlights that adversarially correlated noise can amplify the discrepancy between fidelity and benchmark results, such as linear XEB, which corresponds to the deviation to ergodicity with the scheme function $f(p) = p^2$. Specifically, Ref.~\cite{gao_limitations_2024} demonstrates that errors localized in specific spatial or temporal regions can significantly distort the particle distributions in their diffusion-reaction model. Such inhomogeneous error configurations can artificially elevate benchmarking metrics like linear XEB without a corresponding improvement in the actual circuit fidelity. Similarly, for ergodicity-based fidelity estimation, such adversarial noise configurations can lead to substantial misinterpretations of the system's fidelity. A rigorous analysis of deviation to ergodicity under adversarial noise could be established by employing techniques developed in Ref.~\cite{gao_limitations_2024}.

In summary, while the global depolarizing noise assumption and weakly correlated noise ensure that the deviation to ergodicity is a robust framework for estimating circuit fidelity in the presence of local and weakly correlated noise, its reliability diminishes in adversarial or strongly correlated noise regimes. Extending these methods to address such noise models remains an important open challenge and a promising direction for future research.

\subsection{Experimental Results}

We check the performance of Google's 53-qubit quantum chip\cite{arute_quantum_2019} using our benchmarking scheme.  Using their public data, we calculate the experimental deviation of ergodicity $\widetilde{\mathrm{DE}}_{f}$ for $f(p) = N^i p^i/[(i-1)!(i-1)]$ with $i=2,3,4$ respectively and plotted the results in \cref{fig:DE_vs_number_of_qubits_of_Google}. Here $N^i/[(i-1)!(i-1)]$ is used to normalize the deviation of ergodicity, such that the maximal deviation is 1, enabling fair comparison for different scheme functions. We calculate the experimental deviation of ergodicity for the chip operating with different number of qubits $n$, where each point is calculated using 50,000 samples obtained in Google's experiments. 
More specifically, we use their experimental data obtained by applying $14$ random gate layers of the EFGH pattern with $6$ elided edges (defined in their paper~\cite{arute_quantum_2019}).  By analyzing the data, we find that the deviation of ergodicity (normalized) becomes larger when the number of qubits grows, suggesting the fidelity of the chip drops when more qubits are operating. We also find that the normalized deviation of ergodicity is similar for the three polynomials, which is an evidence suggesting that the noise correlation is weak and does not destroy the relation between the fidelity and the deviation of ergodicity---for otherwise the fidelity estimated using different polynomials can be different.

\subsection{Relation to other results}

We now compare our results with several recent studies that examine noisy random quantum circuits by mapping random circuit sampling to models in statistical mechanics \cite{liu_benchmarking_2021, dalzell_random_2021, deshpande_tight_2021, gao_limitations_2024}. 
These studies provide insights into noisy random quantum circuits and the relation between XEB fidelity and circuit fidelity. Our approach, which utilizes ergodicity, provides an alternative confirmation that XEB fidelity can be a reliable approximation of circuit fidelity, provided that the noise is minor and weakly correlated. 
In Ref.~\cite{dalzell_random_2021}, it was shown that if the local noise is weak and unital, then a random quantum circuit will transform the local noise into a global depolarizing noise, hence justifying the weak correlation assumption. 
With their analysis, the XEB fidelity will approach the depolarization fidelity. 
However, in Ref.~\cite{gao_limitations_2024}, they studied the conditions under which the XEB fidelity agrees with or deviates from the circuit fidelity.
Specifically, they found that if the noise is sufficiently weak and homogeneous, the discrepancy between these two quantities will be small.
But if the noise is highly correlated, their results show that the XEB fidelity can be orders of magnitude higher than the circuit fidelity, challenging the XQUATH conjectured by Aaronson and Gunn~\cite{aaronson_classical_2019}. 
We would like to remark that our analysis in \cref{subsec:case_studies} dose not contradict to their results, as we consider weakly correlated noise. 
Instead, our analysis aligns with previous works~\cite{dalzell_random_2021,gao_limitations_2024} and contributes to this line of research using alternative tools and perspectives.

\section{Conclusion}
\label{sec:conclusion}

We introduce the concept of ergodicity to random circuit sampling and formulate it as a condition where the average of a function of output probabilities over a random circuit ensemble is close to the average of the function over output strings of one circuit instance from the ensemble. 
We proved the ergodicity holds for unitary $2t$-design relative to polynomials with degree at most $t$ and positive coefficients. Our results establish a useful tool in the analysis of random circuit sampling, with a clear characterization of its applicable regime. Based on the ergodicity condition, we propose a benchmarking scheme of quantum chips applicable in the supremacy regimes. 
As case studies, our benchmarking scheme can recover circuit fidelity for global depolarizing noise, and can recover Google's result on linear cross-entropy benchmarking for minor and weakly-correlated noise.

Our study formulates a novel property of random circuit sampling, establishes a notable connection between ergodicity and noise levels in these circuits, and provides innovative insights for quantum benchmarking frameworks in the supremacy regimes. Future investigations can be made to test the experimental performance of the benchmarking scheme via deviation of ergodicity, to compare the results to other benchmarking tools, and to study the impact on the benchmarking of correlated noise. It is also interesting to connect our results for $f(p)=p\ln p$ to quantum thermodynamics and investigate the physical significance of the deviation of ergodicity in terms of thermodynamical concepts, such as free energy and work fluctuation theorems.

\medskip
\textbf{Acknowledgement---}
We thank Xun Gao and Feng Pan for the helpful comments. 
B.C. acknowledges the support by the Sydney Quantum Academy and 
the support by the National Research Foundation, Singapore, and A*STAR under its CQT Bridging Grant and its Quantum Engineering Programme under grant NRF2021-QEP2-02-P05. F. M. acknowledges the support by City University of Hong Kong (Project No. 9610623) and the YTJX academy.
M.H.Y. is supported by National Natural Science Foundation of China (11875160 and U1801661), 
the Natural Science Foundation of Guangdong Province (2017B030308003), 
the Science, Technology and
Innovation Commission of Shenzhen Municipality (JCYJ20170412152620376 and JCYJ20170817105046702 and KYTDPT20181011104202253), 
the Key R\&D Program of Guangdong province (2018B030326001), the Economy, 
Trade and Information Commission of Shenzhen Municipality (201901161512), 
Guangdong Provincial Key Laboratory (Grant No. 2019B121203002).

%%%%%%%%%%%%%%%%%%%%%%%% ref %%%%%%%%%%%%%%%%%%%%%%%%%%%%%%%%

\bibliography{ref}

\onecolumngrid

\appendix

%%%%%%%%%%%%%%%%%%%%%%%%%%%%%%%%%%%%%%%%%%%%%%%%%%%%

\section{Output probabilities of Haar-distributed unitaries}
\label{sec:haar_output_prob}

Let $p = P_U(x_0)$ be the output probability of $x_0$, where $U$ is sampled from Haar measure.
We know that $p$ itself is a random variable, with the sample space being the group of $N \times N$ unitary matrices.
In this section, we show that $p$ obeys the Beta distribution, which tends to the Porter-Thomas distribution when $n \to \infty$.

\subsection{Beta distribution}
\label{subsec:Beta}

We will assume in the following that the quantum circuit $U$ is Haar random. To begin with, it is well established that a Haar random unitary can be decomposed into the components $U = C R$~\cite{Ozols2009}.
Here, $C$ is sampled from the Ginibre ensemble with $C_{jk} := a_{jk} + i b_{jk}$, where $a_{jk}$ and $b_{jk}$ are sampled from the standard normal distribution $N(0, 1)$ independently and $0 \leq j, k \leq N - 1$.
Generally, $C$ is not a unitary, so $R$ is used to perform the Gram-Schmidt process to $C$, making $CR$ a unitary. This implies that $R$ is an upper triangular matrix.
In particular, the first column of $R$ only has one non-zero entry, which is,
\begin{align}
    r_{00} = \frac{1}{\left(\sum_{j=0}^{N-1} |C_{j0}|^2 \right)^{1/2}} = \frac{1}{\left(\sum_{j=0}^{N-1} a_{j0}^2 + b_{j0}^2 \right)^{1/2}} \ .
\end{align}
Below, we will interchangeably use $\ket{x}$ or $\ket{j}$ to denote the computational basis state with $0 \leq j \leq N - 1$. 
Specifically, given an $n$-bit string $x = (x_{n-1}, \cdots, x_0)$, we transform it into $j = \sum_{k=0}^{n-1} x_k 2^{k}$.
Then, the output probability of the $j$-th output $P_U(j) = |\mel{j}{U}{0}|^2$ is the square of the norm of the entry in the first column of $U$, which can be expressed as,
\begin{align}
    P_U(j) = (a_{j0}^2 + b_{j0}^2) r_{00}^2 \ .
\end{align}

We would like to show that $P_U(j)$, as a random variable over all choices of $U$, obeys the Beta distribution.
First, suppose we have $k$ independent standard normal random variables $c_1, \cdots, c_k$, i.e., $c_i \sim N(0, 1)$.
It is well-known that the sum of squares of $c_i$, obeys the chi-squared distribution with $k$ degrees of freedom; formally, we write $c_1^2 + \cdots + c_k^2 \sim \chi^{2}(k)$~\cite{Bailey1992}.
That means, $P_U(j)$ is of the form $\frac{X}{X+Y}$, where $X = a_{j0}^2 + b_{j0}^2 \sim \chi^2(2)$ and $Y = \sum_{k \neq j} a_{k0}^2 + b_{k0}^2 \sim \chi^2(2N-2)$.  
With the relation of chi-squared distribution and Beta distribution~\cite{Bailey1992}, we have
\begin{align}
    P_U(j) \sim \mathrm{Beta}(1, N-1) \ ,
\end{align}
with the probability density function $\Pr(p) = (N-1) (1 - p)^{N-2}$ and $p = P_U(j)$.
Note that $\Pr(p)$ is the probability of the random variable $p = P_U(j)$.

We remark that $P_U(j)$ obeys the same distribution for all $j$, but they are not necessarily independent, since they all contains a common factor $r_{00}^2$, which is itself a random variable. 
But their correlation is small, as indicated by Theorem~1 in the main text.
Therefore, given a randomly chosen $U$, one obtains $2^n$ output probabilities $P_U(j)$.
If we count the frequency of $P_U(j)$ over a small range, then the value is approximately $\Pr(p) \dd{p}$, i.e., the probability associated with the Beta distribution.
Explicitly, we have
\begin{equation}\label{eq:PrPU}
    \frac{|\{ j | P_U(j) \in [p_0, p_0 + \dd{p}] \}|}{N}\simeq\Pr(p_0)\,\mathrm{d}p\ .
\end{equation}
The approximation will be better as the number of qubits $n$ increases.

\subsection{Porter-Thomas properties}
\label{subsec:PT_property}
Porter-Thomas distribution of probability $P_U(j)$ for having outcome $\ket{j}$ is given by
\begin{equation}\label{eq:Porter-Thomas}
    \Pr(p) = N\mathrm{e}^{-Np}\ .
\end{equation}
Consider the expression of Beta distribution under the limit of $N\to\infty$:
\begin{align}
    \lim_{N\to\infty}(N-1)(1-p)^{N-2} &= \lim_{N\to\infty}(N-1)\mathrm{e}^{(N-2)\ln(1-p)}\\
    &= \lim_{N\to\infty}N\mathrm{e}^{N\ln(1-p)}
\end{align}
Apply Taylor expansion to $\ln(1-p)$ at $p = 0$ to the 1st order, we have
\begin{equation}
    \ln(1-p) = -p + o(p)\ ,
\end{equation}
where $o(p)$ is a correction term depends merely on $p$. Considering the limiting case $N\to\infty$, $o(p)$ can be expressed as $O(1)$ since it is independent of $N$. Therefore, we have
\begin{equation}
    (N-1)(1-p)^{N-2} = N\mathrm{e}^{-N(p+O(1))} = N(\mathrm{e}^{-Np}+O(1)),\quad N\to\infty,
\end{equation}
which means that Beta distribution and Porter-Thomas distribution are approximately the same when $N\to\infty$.

\section{The expectation value and the variance for $f(p) = p^i$.}
\label{appendix:calculationforpi}
This section gives the calculation of the expectation value and variance for $f(p)=p^i$. By change of variable $u=Np$, we have
\begin{equation}
     \mathbb{E}_{U \in \Haar} [f(P_U(x_0))] = \int_0^1 f(p) \Pr(p) \dd{p} = \int_0^1 p^i N e^{-N p} \dd p = \frac{1}{N^i} \int_0^N u^i e^{-u} \dd u = \frac{1}{N^i}[\int_0^\infty u^i e^{-u} \dd u -\int_N^\infty u^i e^{-u} \dd u] \, . 
\end{equation}
The term $\int_0^\infty u^i e^{-u} \dd u$ gives us a gamma function, $\Gamma(i+1) = i!$. Now we evaluate the second term $\int_N^\infty u^i e^{-u} \dd u$. By change of variable $u' = u - N$, we have
\begin{equation}
    \int_N^\infty u^i e^{-u} \dd u = \int_0^\infty (u'+N)^i e^{-(u'+N)} \dd u' = \frac{N^i}{e^N} \int_0^\infty (1 + \frac{u'}{N})^i e^{-u'} \dd u \approx \frac{N^i}{e^N} \int_0^\infty (1 + i \frac{u'}{N}) e^{-u'} \dd u = \frac{N^i}{e^N} [e^{-1} + \frac{2i}{N}] \, .
\end{equation}
Therefore,
\begin{equation}
     \mathbb{E}_{U \in \Haar} [f(P_U(x_0))] =\frac{1}{N^i}[i! - \frac{N^i}{e^N} [e^{-1} + \frac{2i}{N}]]= \frac{i!}{N^i} + e^{-\Omega(N)} \, .
\end{equation}
Similarly, $\sigma^2_f$ can be calculated,
\begin{equation}
  \sigma_f^2 = \int_0^1 f(p)^2 N e^{-N p} \dd p  - \mathbb{E}_{U \in \Haar} [f]^2\, = \frac{(2i)!-(i!)^2}{N^{2i}}  + e^{-\Omega(N)} \, . 
\end{equation}
Therefore, the standard deviation $\sigma_f$ is
\begin{equation}
    \sigma_f = \frac{\sqrt{(2i)!-(i!)^2}}{N^{i}} + e^{-\Omega(N)} \, .
\end{equation}

\section{Deviation of ergodicity for completely noisy chips and for global depolarizing noise.}
\label{appendix:exampleofDeviationErgodicityforCompletely_Noisy_Chips}

\paragraph{Completely noisy chips.}
Now we analyze the deviation of ergodicity when the quantum chip is completely random, such that $Q_U(x) = \frac{1}{N}$ for any $x$ and any $U$.

For $f(p) = N^i p^i$ with $i>1$, we have 
\begin{equation}
        C_f(P_U, Q_U) = \frac{1}{N} \sum_{x \in \{0, 1\}^n} \frac{f(P_U(x))}{P_U(x)} Q_U(x)  = \frac{1}{N^2} \sum_{x \in \{0, 1\}^n} \frac{f(P_U(x))}{P_U(x)} =  N^{i-2} \sum_{x \in \{0, 1\}^n} {P^{i-1}_U(x)}  .
\end{equation}
In order to evaluate $N^{i-2} \sum_{x \in \{0, 1\}^n} {P^{i-1}_U(x)}$, observe that the ergodicity holds for $f(p) = p^{i-1}$ and we have
\begin{equation}
\frac{1}{N} \sum_{x\in\{0,1\}^n}P^{i-1}_U(x) =\mathbb{E}_{U \in \Haar} [P^{i-1}_U(x)] \pm O\left( \frac{\sigma_f}{\sqrt{N}} \right) = \frac{(i-1)!}{N^{i-1}} \pm O\left( \frac{\sqrt{(2i-2)!-[(i-1)!]^2}}{N^{i-1}\sqrt{N}} \right) .
\end{equation}
Thus, the correlation $C_f(P_U, Q_U)$ is
\begin{equation}\label{eq:C_fCompletelyNoisy}
    C_f(P_U, Q_U) =   (i-1)! \pm O\left( \frac{\sqrt{(2i-2)!-[(i-1)!]^2}}{\sqrt{N}} \right)\ . 
\end{equation}

Therefore, the completely noisy quantum chip violates the ergodicity condition for $f(p) = N^i p^i$,
\begin{equation}
 \mathrm{DE}_f= \abs{\EE_U [f(P_U)] -  C_f(P_U, Q_U)} = (i-1)!(i-1) \pm O\left(\frac{\sqrt{(2i-2)!-[(i-1)!]^2}}{\sqrt{N}} \right) >  O\left(\frac{\sigma_f}{\sqrt{N}}\right) = O\left(\frac{\sqrt{(2i)!-(i!)^2}}{\sqrt{N}}\right) \, .
\end{equation}

\paragraph{Global depolarizing noise.}
Suppose that the quantum chip undergoes a global depolarizing channel before the measurement, so that the output state $\rho_U$ is given by,
\begin{equation}
    \rho_U = F \ket{\psi_U}\bra{\psi_U} + (1-F) \frac{\mathbb{I}}{N} \, ,
\end{equation}
where $\ket{\psi_U} = U \ket{0}^{\otimes n}$ is the ideal output of the noiseless quantum circuit $U$ and $F$ is equal to the circuit fidelity $\mel{\psi_U}{\rho_U}{\psi_U}$ up to a additive factor at most $1/N$. 
Then, the experimental probability is given by
\begin{equation}
    Q_U(x) = F P_U(x) + (1-F) \frac{1}{N} \, .
\end{equation}
\cref{thm:ergodicity_polynomials} establishes that when $Q_U(x) = P_U(x)$ for all $x$, we have
\begin{equation}\label{eq:C_f_P_UP_U}
    C_f(P_U,P_U) =\mathbb{E}_{U \in \Haar} [f(P_U(x))] \pm O\left( \frac{\sigma_f}{\sqrt{N}} \right) = i! \pm O\left( \frac{\sqrt{(2i)!-[i!]^2}}{\sqrt{N}} \right) .
\end{equation}
Thus, the linearity of $    C_f(P_U, Q_U)$ in $Q_U$ together with \cref{eq:C_f_P_UP_U} and \cref{eq:C_fCompletelyNoisy} establish the following result, 
\begin{equation}
        C_f(P_U, Q_U) =F i! + (1-F) (i-1)! + O \left(\frac{1}{\sqrt{N}}\right) \, .
\end{equation}

\section{Analysis on the ergodicity condition}
\label{sec:analysis_random_unitary_condition}

We will prove the following proposition in this section.
\begin{proposition}
Let $q_1, q_2$ be two positive numbers, and $x\neq y$, then 
\begin{equation}
    \Cov \left( P^{q_1}_U(x), P^{q_2}_U(y) \right) < 0 \, ,
\end{equation}
where $P^{q_1}_U(x) = |\bra{x}U\ket{0}|^{2 q_1}$ and  $P^{q_2}_U(y)$ is defined similarly.
\end{proposition}
\begin{proof}
By definition, the covariance is 
\begin{equation}
     \Cov \left( P^{q_1}_U(x), P^{q_2}_U(y) \right) = \EE_U \left[ P^{q_1}_U(x) P^{q_2}_U(y) \right]- \EE_U \left[ P^{q_1}_U(x) \right] \EE_U \left[P^{q_2}_U(y) \right] \ .
\end{equation}
The expectation here takes the  following integral form in terms of Haar random unitaries,
\begin{equation}
    \EE_U \left[ P^{q_1}_U(x) P^{q_2}_U(y) \right] = \int  |\bra{x}U\ket{0}|^{2 q_1} |\bra{y}U\ket{0}|^{2 q_2} \dd U .
\end{equation}
Such integral can be easily evaluated using results shown in Ref.~\cite{puchala2011symbolic}. It is proved that 
\begin{equation}
    \int \prod_{i=1}^N |\bra{i}U\ket{0}|^{2 q_i} = \Gamma(N) \frac{\Gamma(q_1+1)\Gamma(q_2+1)...\Gamma(q_N+1)}{\Gamma(q_1 + q_2 +...+q_N +N)} \, ,
\end{equation}
where $q_i$'s are non-negative numbers. Applying this result, we have
\begin{equation}
    \EE_U \left[ P^{q_1}_U(x) P^{q_2}_U(y) \right] = \Gamma(N) \frac{\Gamma(q_1+1)\Gamma(q_2+1)}{\Gamma(q_1+q_2+N)} \, ,
\end{equation}
and 
\begin{equation}
    \EE_U \left[ P^{q_1}_U(x) \right] \EE_U \left[ P^{q_2}_U(y) \right] = \Gamma(N)^2 \frac{\Gamma(q_1+1)\Gamma(q_2+1)}{\Gamma(q_1+N)\Gamma(q_2+N)} \, .
\end{equation}
Then, it is easy to check that 
\begin{align}
    &\Cov \left( P^{q_1}_U(x), P^{q_2}_U(y) \right)\\
    &= \Gamma(q_1+1)\Gamma(q_2+1)\left[\frac{\Gamma(N)}{\Gamma(q_1+q_2+N)} - \frac{\Gamma(N)\Gamma(N)}{\Gamma(q_1+N)\Gamma(q_2+N)}\right] \\
    &= \Gamma(q_1+1)\Gamma(q_2+1)\left\{\frac{1}{(q_1+q_2+N-1)(q_1+q_2+N-2)...N} - \frac{1}{[(q_1+N-1)...N] [(q_2+N-1)...N]} \right\} \\
    &<0
\end{align}
\end{proof}

With this proposition, the following corollary is immediate.
\begin{corollary}
Let $f(p)=\sum_i a_i p^i$ be a polynomial with non-negative coefficients $a_i$, then
\begin{equation}
    \Cov \left( f(P_U(x)), f(P_U(y)) \right) < 0 \, ,
\end{equation}
for any $x\neq y$.
\end{corollary}

In particular, we proved that the covariance is negative for $f(p)=p$ and $f(p)=p^2$.
The above implies that the random unitary condition holds for the polynomial function $f(p) = \sum_i a_i p^i$ where $a_i$'s are all positive, as we discussed in the main text.

\section{Proof of ergodicity for cross entropy}\label{appendix:ergodicity_proof_plogp}

Here, we are going to apply the replica trick to prove \cref{thm:ergodicity_log}.
Recall that the replica trick says $\ln p = \lim_{i \to 0} \frac{p^i - 1}{i}$ and our goal is to prove ergodicity of Haar random unitaries relative to the function $f(p) = p \ln{p}$.

Let $g_i(p) = \frac{p^{i+1} - p}{i}$ and we have $f(p) = \lim_{i \to 0} g_i(p)$.
We want to check the covariance
\begin{align}
    \Cov \left( g_i(P_U(x)), g_i(P_U(y)) \right) &= \frac{1}{i^2} \Cov \left( P_U(x)^{i+1} - P_U(x), P_U(y)^{i+1} - P_U(y) \right) \\
    &= \frac{1}{i^2} \left[ \Cov \left( P_U(x)^{i+1}, P_U(y)^{i+1} \right) - \Cov \left( P_U(x), P_U(y)^{i+1} \right) - \Cov \left( P_U(x)^{i+1}, P_U(y) \right) + \Cov \left( P_U(x), P_U(y) \right) \right] 
\end{align}
is less than or equal to zero or not.
Applying \cref{lemma:covariance_over_U}, we have 
\begin{equation}\label{eq:COVbeforelog}
    \Cov \left( g_i(P_U(x)), g_i(P_U(y)) \right) = \frac{\Gamma[N]}{i^2}\left[  \frac{\Gamma (i+2)\Gamma (i+2)}{\Gamma (2 i+N+2)}-\frac{2 \Gamma (i+2)}{\Gamma (i+N+2)}-\frac{\Gamma (N) \left[ \Gamma (i+N+1)-\Gamma (i+2) \Gamma (N+1) \right]^2}{\Gamma (N+1)^2 \Gamma (i+N+1)^2}+\frac{1}{\Gamma (N+2)}\right]
\end{equation}
Taking the limit $i \to 0$, we have
\begin{align}
    \Cov(f(P_U(x)), f(P_U(y))) = \lim_{i \to 0} \Cov \left( g_i(P_U(x)), g_i(P_U(y)) \right) \ ,
\end{align}
which means that 
\begin{align} \label{eq:COVlog}
    \Cov(f(P_U(x)), f(P_U(y))) = &\frac{6 [\psi ^{(0)}(N+1)]^2+12 (\gamma -1) \psi ^{(0)}(N+1)-6 \psi ^{(1)}(N+1)+6 \gamma ^2+\pi ^2-12 \gamma }{6 N^2} \nonumber\\
    &-\frac{\left(12 [\psi ^{(0)}(N+1)]^2+24 (\gamma -1) \psi ^{(0)}(N+1)-6 \psi ^{(1)}(N+1)+12 \gamma ^2+\pi ^2-24 \gamma +6\right)}{6 N^2} \nonumber\\
    & +\frac{2 (\gamma -1) \psi ^{(0)}(N+2)+\psi ^{(0)}(N+2)^2-\psi ^{(1)}(N+2)+(\gamma -1)^2}{(N+1)N} \,  ,
\end{align}
where $\psi^{(m)}(z):= \frac{\dd^{m+1}}{\dd z^{m+1}} \ln(\Gamma(z))$ is the polygamma function of order $m$ and $\gamma$ is the Euler constant. 
Rearranging \cref{eq:COVlog} gives
\begin{align}
    \Cov(f(P_U(x)), f(P_U(y))) = \frac{-\ln ^2(N)-2 \gamma  \ln (N)+4 \ln (N)-\gamma ^2+4 \gamma -4}{N^3}+O\left(\frac{1}{N^4}\right) \, ,
\end{align}
which is smaller than $0$ for sufficiently large $N$.
Finally, we use similar procedure as in the proof of \cref{thm:ergodicity_polynomials} to conclude the proof for the ergodicity relative to $f(p) = p \ln p$.

\section{Discussion on Levy's lemma}
\label{sec:levy_lemma_discussion}

Levy's Lemma is a probabilistic result in high-dimensional geometry which states that a Lipschitz continuous function on a high-dimensional sphere is close to its mean value with high probability~\cite{milman1986asymptotic,ledoux2001concentration}. Levy's lemma has been widely used in quantum information to demonstrate that most quantum states are nearly maximally entangled~\cite{hayden2006aspects} and to justify the typicality arguments in the study of large quantum systems~\cite{hu2019generalized,dahlsten2014entanglement}, providing a foundation for understanding phenomena such as the concentration of measure in the space of quantum states~\cite{muller2011concentration,mcclean2018barren}.

One may be tentative to think that the ergodicity condition is a direct consequence of Levy's lemma.
Here, we discuss the relation between Levy's lemma and the ergodicity condition and show that these two are not equivalent.
We first recall Levy's lemma.

\begin{lemma}[Levy's Lemma, see \cite{milman1986asymptotic}]
    Let $g$ be a function that maps vectors on the $k$-dimensional unit sphere $\mathbb{S}^k$ to real numbers. 
    Suppose $g$ has Lipschitz constant $\eta$, meaning that $|g(X_1)-g(X_2)| \leq \eta \norm{X_1 - X_2}$  with respect to the Euclidean norm for any $X_1, X_2 \in \mathbb{S}^k$. Then there exists a constant $C$ such that for any $X\in \mathbb{S}^k$ chosen uniformly at random, one has
    \begin{equation}\label{eq:levy_lemma}
        \Pr[|g(X)-\EE_{\mathbb{S}^{2N}} [g(X)]| > \epsilon] \leq 2 e^{-C (k+1) \epsilon^2 /\eta^2} \, .
    \end{equation}
\end{lemma}

For any quantum state $U\ket{0^n}$, we define a random variable $X_U \in \mathbb{S}^{2N}$ by concatenating the real and imaginary parts of $U \ket{0^n}$.
When $U$ is Haar random, $X_U$ is uniformly distributed over $\mathbb{S}^{2N}$.
Define a function $g: \mathbb{S}^{2N} \to \mathbb{R}$ as 
\begin{align}
    g(X_U) := \frac{1}{N} \sum_{x} f(P_U(x)) \ .
\end{align}
Then, $\EE_{\mathbb{S}^{2N}} [g(X_U)] = \frac{1}{N} \sum_x \EE_U [f(P_U)] = \EE_U [f(P_U)]$.
From Levy's lemma, we have
\begin{align}
    \Pr(\left| \EE_U [f(P_U)] - \frac{1}{N} \sum_x f(P_U(x)) \right| \geq \alpha \frac{\sigma_f}{\sqrt{N}}) \leq A \ ,
\end{align}
where $A := 2 \exp(- \frac{C (2N+1) \alpha^2 \sigma_f^2}{N \eta^2})$ and $\alpha$ is a positive constant.
Below, we take $f(p) = p^k$, which gives $\sigma_f^2 = \Theta (1/N^{2k})$.

To examine the scaling of $A$, we need to estimate the Lipschitz constant $\eta$ of $g(X_U)$. 
When $U = I^{\otimes n}$, we have $g(X_{I^{\otimes n}}) = 1/N$ and when $U = H^{\otimes n}$, we have $g(X_{H^{\otimes n}}) = 1/N^k$.
The distance between $X_{I^{\otimes n}}$ and $X_{H^{\otimes n}}$ as two points in $\mathbb{S}^{2N}$ is upper bounded by $\sqrt{2}$. 
Therefore, we have
\begin{equation}
    \eta \geq \frac{\abs{g(X_{I^{\otimes n}}) -g(X_{H^{\otimes n}})}}{\norm{X_{I^{\otimes n}} - X_{H^{\otimes n}}}} \geq \frac{1/N - 1/N^k}{\sqrt{2}} \, .
\end{equation}
Putting together, we have $A = \Omega \left( \exp(-\frac{\alpha^2}{N^{2k-2}}) \right)$, which is away from zero.
In contrast, for $f(p) = p^k$, \cref{thm:ergodicity_polynomials} gives a bound of $1/\alpha^2$ for the probability of deviation.
Therefore, Levy's lemma does not appear to imply the ergodicity condition.

\end{document}